\documentclass[10pt,journal,twocolumn,letterpaper]{IEEEtran}
\ifCLASSINFOpdf
\else
\fi
\usepackage{psfrag, epsfig, amsmath,amssymb,amsthm}

\newtheorem{proposition}{Proposition}
\newtheorem{theorem}{Theorem}

\newtheorem{hypo}{Hypothesis}

\newtheorem{remark}{Remark}

\begin{document}
\title{Coordination in Network Security Games:\\
a Monotone Comparative Statics Approach}

\author{\IEEEauthorblockN{Marc Lelarge}
\IEEEauthorblockA{INRIA - ENS\\
Paris, France\\
Email: marc.lelarge@ens.fr}}


%Telephone: (800) 555--1212, Fax: (888) 555--1212}

%\IEEEspecialpapernotice{}

\maketitle

\begin{abstract}
%\boldmath
Malicious softwares or malwares for short have become a major security
threat. While originating in criminal behavior, their impact are also
influenced by the decisions of legitimate end users. Getting agents in
the Internet, and in networks in general, to invest in and deploy
security features and protocols is a challenge, in particular because
of economic reasons arising from the presence of network externalities.

In this paper, we focus on the question of incentive alignment for agents of a large
network towards a better security.
We start with an economic model for a single agent, that
determines the optimal amount to invest in protection. The model takes
into account the vulnerability of the agent to a security breach and
the potential loss if a security breach occurs. We derive conditions on the quality of the protection to ensure that the
optimal amount spent on security is an increasing function of the
agent's vulnerability and potential loss. We also show that for a
large class of risks, only a small fraction of the expected loss
should be invested.

Building on these results, we study a network of interconnected
agents subject to epidemic risks.
We derive conditions to ensure that the incentives of all agents are
aligned towards a better security. 
When agents are strategic, we show that security investments are always socially
inefficient due to the network externalities. Moreover alignment of
incentives typically implies a coordination problem, leading to an equilibrium with a very high price of
anarchy.\footnote{extended abstract of this work presented at INFOCOM
  2012. This version corrects some inaccuracies
  of \cite{lelinfocom12}. The author wishes to thank the anonymous reviewers for valuable comments.}
\end{abstract}


% \fi
%\IEEEpeerreviewmaketitle

\section{Introduction}

Negligent users who do not protect their computer by regularly
updating their antivirus software and operating system are clearly
putting their own computers at risk. But such users, by connecting to
the network a computer which may become a host from which viruses can
spread, also put (a potentially large number of) computers on the
network at risk \cite{a01,at08}.
This describes a common situation in the Internet and in enterprise networks, in which
users and computers on the network face {\em epidemic risks}.
Epidemic risks are risks which depend on the behavior of
other entities in the network, such as whether or not those entities invest in
security solutions to minimize their likelihood of being infected.
\cite{ocde} is a recent OECD survey of the misaligned incentives as perceived by multiple stake-holders.
Our goal in this paper is to study conditions for alignment of
incentives for agents of a large
network subject to epidemic risks and its implications for the equilibria.

Our work allows a better understanding of economic network
effects: there is a {\em total effect} if one agent's adoption of a protection
benefits other adopters and there is a {\em marginal effect} if it increases others'
incentives to adopt it (we refer to Section 3 of \cite{fk07} for a
comprehensive survey about network effects).
In information security economics, the presence of the total effect has been
the focus of various recent works starting with Varian's work
\cite{var02}. When an agent protects itself, it benefits not only to those who
are protected but to the whole network.
Indeed there is also an incentive to free-ride the total effect. 
Those who invest in self-protection incur some cost and in
return receive some individual benefit through the reduced individual
expected loss. But part of the benefit is public: the reduced indirect
risk in the economy from which everybody else benefits.
As a result, the agents invest
too little in self-protection relative to the socially efficient
level. The efficiency loss (referred to as the price of anarchy) has been
quantified in various game-theoretic models
\cite{kunreuther03,sig08,gcc08,law08,bambos,oovm}.

In this paper, we focus on the marginal effect and its relation to the
coordination problem (see Section 3.4 in \cite{fk07}).
To understand the mechanism of incentives regarding security in a
large network, we need to analyze how an increase in the total
population adopting security will impact one agent's incentive to
adopt it. 
To do so, we use a monotone comparative statics approach and start with an economic model for a single agent that determines the
optimal amount to invest in protection. We follow the approach
proposed by Gordon and Loeb in \cite{gl02}. They found that the
optimal expenditures for protection of an agent do not
always increase with increases in the vulnerability of the agent.
Crucial to their analysis is the security breach probability function
which relates the security investment and the vulnerability of the
agent with the probability of a security breach after protection.
This function can be seen as a proxy for the quality of the security
protection. Our first main result (Theorem \ref{th:monot1}) gives
sufficient conditions on this function to ensure that the optimal
expenditures for protection always increase with increases in the
vulnerability of the agent (this sensitivity analysis is called \emph{monotone comparative statics} in economics). From an economic perspective, these
conditions will ensure that all agents with sufficiently large
vulnerability value the protection enough to invest in it.
We also extend a result of \cite{gl02} and show (Theorem \ref{th:1/e})
that if the security breach probability function is log-convex in the
investment, then a {\em risk-neutral}\footnote{i.e an agent indifferent to investments
that have the same expected value: such an agent will have no
preference between i) a bet of either 100\$ or nothing, both with a
probability of 50\% and ii) receiving 50\$ with certainty} agent never
invests more than 37\% of the expected loss.

Building on these results, we study a network of interconnected agents
subject to epidemic risks. We model the effect of the network through
a parameter $\gamma$ describing the information available to the agent and
capturing the security state of the network. We show that our
general framework extends previous work \cite{sig08,lel09} and allows to
consider a security breach probability function depending on the
parameter $\gamma$ and possibly other private informations on the
vulnerability of the asset.
Our third main result (Theorem \ref{th:gen}) gives sufficient
conditions on this function to ensure that the optimal protection
investment always increases with an increase in the security state of
the network.

This property will be crucial in our last analysis:
we use our model of interconnected agent in a game theoretic
setting where agents anticipate the effect of their actions on the
security level of the network.
We diverge form most of the literature on security games (some
exceptions are \cite{jgcc10,sig08,leboudec}) and relax the complete
information assumption, i.e. each player's security breach probability
is not common knowledge but instead a private information. 
In our model only global statistics are publicly available and agents do not disclose any
information concerning their own security strategy.

We show how the monotonicities (or the lack of
monotonicities) impact the equilibrium of the security game.
In particular, alignment of incentives typically implies a
coordination problem, leading to an equilibrium with a very high price
of anarchy.
Moreover, we distinguish two parts in the network
externalities that we call public and private.
Both types of externalities are positive since any additional agent investing in security will increase the security level of the whole network. However, the effect of this additional agent will be different for an agent who did not invest in security from an agent who already did invest in security.
The public externalities correspond to the network effect on insecure
agents while the private externalities correspond to the network
effect on secure agents (also called total effect in the economics literature \cite{fk07}).

As a result of this separation of externalities, some surprising phenomena can occur: 
also both externalities are positive, there are situations where the incentive to invest in
protection decreases as the fraction of the population investing
in protection increases. This is
an example where the total effect holds but the marginal effect
fails (which is essentially a case where Segal's increasing
externalities \cite{segal} or Topkis'supermodularity \cite{topbook}
fails).
We also show that in the security
game, security investments are always inefficient due to the network
externalities. This raises the question whether economic tools like
insurance \cite{blinfocom08,blweis08,lbinfocom09} could be used to
lower the social inefficiency of the game\footnote{Note that in this case the risk-neutral assumption made in this paper should be replaced by a risk-adverse assumption.}?

The rest of the paper is organized as follows. In Section 
\ref{sec:gl}, the optimal security investment for a single agent is
analyzed. In Section \ref{sec:inter}, we extend it to an
interconnected agent and show it connects with the epidemic risk
model. Finally in Section \ref{sec:equ}, we consider the case where
agents are strategic. We introduce the notion of fulfilled
expectations equilibrium and show our main game theoretic results.

\section{Optimal security investment for a single agent}\label{sec:gl}

In this section, we present a simple one-period model of an agent
contemplating the provision of additional security to protect a given
information set introduced by Gordon and Loeb in \cite{gl02}. In
one-period economic models, all decisions and outcomes occur in a
simultaneous instant. Thus dynamic aspects are not considered.

\subsection{Economic model of Gordon and Loeb}

The model is characterized by two parameters $\ell$ and $v$ (also
Gordon and Loeb used a bit more involved notation).
The parameter $\ell$ represents the monetary loss caused by a security
breach. The parameter $\ell\in \mathbb{R}_+$ is a positive real number.
The parameter $v$ represents the probability that without additional
security, a threat results in the information set being breached and
the loss $\ell$ occurs. The parameter $v$ is called the \emph{vulnerability} of the asset.
Being a probability, it belongs to the interval $[0,1]$.

An agent can invest a certain amount $x$ to reduce the probability of
loss to $p(x,v)$. We make the assumptions $p(0,v) = v$ and
since $p(x,v)$ is a probability we assume that for all $x>0$ and
$v\in [0,1]$ we have $0\leq p(x,v)\leq v$. The function $p(x,v)$ is called the \emph{security breach probability}.

The expected loss for an amount $x$ spent on security is given by $\ell p(x,v)$.
Hence if the agent is risk neutral, the optimal security investment should
be the value $x^*$ minimizing
\begin{eqnarray}
\label{eq:gl}\min\left\{\ell p(x,v)+x :x\geq 0\right\}.
\end{eqnarray}

We define the set of optimal security investment by
%\begin{eqnarray*}
$\varphi(v,\ell) = \arg \min\left\{ \ell p(x,v)+x: x\geq 0 \right\}$.
%\end{eqnarray*}
Clearly in general the function $\varphi$ is set-valued and we will
deal with this fact in the sequel. For now on, assume that the
function $\varphi$ is real-valued, i.e. sets reduce to singleton.
As noticed in \cite{gl02}, it turns out that the function $\varphi(v,\ell)$
does not need to be non-decreasing in $(v,\ell)$ for general functions
$p(x,v)$. An example given in \cite{gl02} is $p_{GL}(x,v)= v^{\alpha x+1}$,
where the parameter $\alpha>0$ is a measure of the productivity of
information security.
This class of security breach probability functions has the property
that the cost of protecting highly vulnerable information sets becomes
extremely expensive as the vulnerability of the information set
becomes very close to one. This is not the only class of security
breach functions with this property. Their simplicity allows to gain
further insights into the relationship between vulnerability and
optimal security investment.

Indeed, an interior minimum $x^*>0$ is characterized by
the first-order condition:
\begin{eqnarray}
\label{eq:1od}\ell \frac{\partial p}{\partial x}(x^*,v)=-1.
\end{eqnarray}
In the particular case where $p_{GL}(x,v)= v^{\alpha
  x+1}$, we obtain
$\frac{\partial p_{GL}}{\partial x}(x,v)=(\alpha\log v )v^{\alpha x+1}$.
So that solving Equation (\ref{eq:1od}), we get
$\varphi_{GL}(v,\ell) = \frac{-\log\left( -\ell\alpha\log v\right)}{\alpha\log v} -\frac{1}{\alpha}$.

\begin{figure}[htb]
\begin{center}
\includegraphics[angle=0,width=4cm]{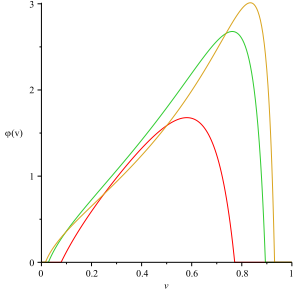}
\hspace{25pt} \caption{Function $\varphi_{GL}(v,\ell)$ as a function of the
  vulnerability $v$ and with parameters: $\ell=10$
and $\alpha=0.5,1,1.5$ (red, green, brown)} \label{fig:gl}
\end{center}\end{figure}

Figure \ref{fig:gl} shows the optimal security investment for various
values of $\alpha$ and $\ell$ as a function of the vulnerability
$v$. In particular, we see that the optimal investment is zero for low
values of the vulnerability and also for high values of the
vulnerability. In other words, in this case, the marginal benefit from
investment in security for low vulnerability information sets does not
justify the investment since the security of the information set is
already good. However if the information set is extremely vulnerable,
the cost of security is too high to be 'profitable', in the sense that
there is no benefit in protecting it. 

\subsection{Sufficient conditions for monotone investment}

In this section, we derive sufficient conditions on the probability
loss in order to avoid the non-monotonicity in the vulnerability of
the information set. In such a case, the information security
decision is simple since there is an augmenting return of investment
with vulnerability: the security manager needs to adjust the security
investment to the vulnerability. Also the security provider should
set the price of its solution so as to remain in a region where such
monotonicity is valid.

First we need to define the monotonicity of a set-valued function. We
say that the set-valued function $f:\mathbb{R}^n\to 2^{\mathbb{R}}$ is
non-decreasing if for any $x^L,x^H\in \mathbb{R}^n$ with $x^L\leq x^H$ (for the product order), we have for any $y^L\in f(x^L)$ and any $y^H\in f(x^H)$: $y^L\leq y^H$.

We start with a particular case (its proof will follow from our more
general result and is given at the end of this section):
\begin{proposition}\label{prop:deriv}
Assume that the function $p(x,v)$ is twice continuously
differentiable on $\mathbb{R}_+\times [0,1]$. If
\begin{eqnarray}
\label{cond} \frac{\partial p}{\partial x}(x,v) \leq 0 , &\mbox{ and, }&
\frac{\partial^2 p}{\partial x\partial v} (x,v)\leq 0
\end{eqnarray}
then the function $(v,\ell) \mapsto \varphi(v,\ell)$ is non-decreasing in $(v,\ell)$.
\end{proposition}
\begin{remark}
%\begin{itemize}
%\item 
The first condition requires that the function
$p(x,v)$ is non-increasing in $x$, i.e. the probability of a security
break is lowered when more investment in security is done.
%\item 
In the particular case of $p_{GL}$ described above, we have
$\frac{\partial^2 p_{GL}}{\partial x \partial v}(x,v) = \alpha v^{\alpha
x}\left(1+\alpha(\alpha x+1)\log v\right)$.
In particular $\frac{\partial^2 p_{GL}}{\partial x \partial v}(x,1) =
\alpha>0$ and we see that the function $p_{GL}$ does not satisfy the
conditions of the proposition which is in agreement with the fact that
the associated function $\varphi_{GL}$ is not monotone in $v$.
%\end{itemize}
\end{remark}

It turns out that we often need to deal with cases where the choice
sets are discrete. In reality, discrete investments in
new security technologies are often more natural, resulting in
discontinuities. For example the amount $x$ could live in a space
$X\subset \mathbb{R}_+$ having empty interiors. In these cases,
Proposition \ref{prop:deriv} is useless. In order to extend it, we
introduce the notion of general submodular functions (see Topkis \cite{to78}). 
We first define the two operators $\wedge$ and $\vee$ in
$\mathbb{R}^n$:
\begin{eqnarray*}
x \wedge y &=& \sup\{ t\in \mathbb{R}^n,\: t\leq x;\: t\leq
y\}\mbox{and,}\\
x \vee y &=& \inf\{ t\in \mathbb{R}^n,\: t\geq x;\: t\geq y\}.
\end{eqnarray*}
A set $S\subset \mathbb{R}^n$ is a lattice if for any $x$ and $y$ in
$S$, the elements $x\wedge y$ and $x\vee y$ are also in $S$.
A real valued function $f$ on a lattice $S$ is submodular if for all
$x$ and $y$ in $S$, 
$f(x\wedge y) + f(x\vee y) \leq f(x)+f(y)$.
$f$ is strictly submodular on $S$ if the inequality is strict for all
pairs $x,y$ in $S$ which cannot be compared with respect to $\geq $,
i.e such that neither $x\geq y$ nor $y\geq x$ holds.

We are now ready to state our main first result which is an adaptation of Theorem 6.1 in \cite{to78}:
\begin{theorem}\label{th:monot1}
Let $S=[0,1]\times \mathbb{R}_+$.
If the function $f:X\times S \to \mathbb{R}$ %defined by $f(x,v,\ell)
                                %= \ell p(x,v)+x$ 
is strictly submodular in the variables
$x$ and $v$ in $X\times [0,1]$ for any fixed $\ell$ and in the
variables $x$ and $\ell$ in $X\times \mathbb{R}_+$ for any fixed $v$, 
then $\varphi(v,\ell)=\arg\min\{ f(x,v,\ell):\:x\in X\}$ is non-decreasing.
\end{theorem}
\begin{remark}
Note that this Theorem does not require to take $f(x,v,\ell)= \ell
p(x,v)+x$. In particular it can also be applied to the case of
risk-adverse agents in which case $f$ depends on the (concave) expected utility
function of the agent. 
\end{remark}
\begin{proof}
If $x\leq x'$ and $x\neq x'$, then $x<x'$ is written.
By the definition of strict submodularity, we see that we have for $x'>x$ and
$(v',\ell')>(v,\ell)$:
\begin{eqnarray*}
f(x',v',\ell')+f(x,v,\ell')&<& f(x',v,\ell')+f(x,v',\ell')\\
f(x',v,\ell')+f(x,v,\ell)&<& f(x',v,\ell)+f(x,v,\ell'),
\end{eqnarray*}
so that we get
\begin{eqnarray*}
f(x',v',\ell')+f(x,v,\ell)&<&f(x',v,\ell)+f(x,v',\ell').
\end{eqnarray*}
This shows that $f$ has strictly increasing
differences in $(x,v,\ell)$, i.e. $f(x,v,\ell)-f(x,v',\ell')$ is
strictly increasing in $x$ for all $(v',\ell')>(v,\ell)$.

Consider $(v',\ell')>(v,\ell)$ and we now show that $y\geq x$ for $y\in
\varphi(v',\ell')$ and $x\in \varphi(v,\ell)$. Suppose that $x>y$, so
that $x\vee y>y$.
Since $y\in \varphi(v',\ell')$ and $x\in \varphi(v,\ell)$, we have
\begin{eqnarray*}
f(x\vee y,v',\ell') &\geq& f(y,v',\ell') \mbox{ and, }\\ 
f(x\wedge y,v,\ell)&\geq& f(x,v,\ell).
\end{eqnarray*}
Using the fact that $f$ has strictly increasing differences, and
$x\vee y>y$, we get:
\begin{eqnarray*}
f(x\vee y,v',\ell')-f(y,v',\ell')&<&f(x\vee y,v,\ell)-f(y,v,\ell).
\end{eqnarray*}
By the definition of submodularity, we have:
\begin{eqnarray*}
f(x\vee y,v,\ell)-f(y,v,\ell)&\leq &f(x,v,\ell)-f(x\wedge y,v,\ell)
\end{eqnarray*}
Hence we finally get:
\begin{eqnarray*}
0&\leq& f(x\vee y,v',\ell')-f(y,v',\ell')\\
&<&f(x,v,\ell)-f(x\wedge y,v,\ell)\leq 0,
\end{eqnarray*}
which provides the desired contradiction.
\end{proof}
\begin{remark}\label{rem:sub}
It follows from the proof, that the sufficient conditions on $f$ to
insure that $\varphi$ is non-decreasing, are equivalent
to:
$f(x,v,\ell)-f(x,v',\ell')$ is
strictly increasing in $x$ for all $(v',\ell')>(v,\ell)$.
\end{remark}
\begin{proof}{of Proposition \ref{prop:deriv}:}\\
It follows from the definition of submodularity, that if $f$ is
twice-continuously differentiable, then $\frac{\partial^2 f}{\partial x\partial v} (x,v,\ell)\leq 0$ implies that $f$ is strictly submodular in the variables
$x$ and $v$ in $X\times [0,1]$ for any fixed $\ell$.
Taking, $f(x,v,\ell)= \ell p(x,v)+x$, we get $\frac{\partial^2 f}{\partial x\partial v} (x,v,\ell)=\ell
\frac{\partial^2 p}{\partial x\partial v} (x,v)$, we get one of the
condition of Proposition \ref{prop:deriv}. The other condition comes
from the symmetric condition on $f$: 
$\frac{\partial^2 f}{\partial x\partial \ell} (x,v,\ell)\leq 0$.
\end{proof}

\subsection{A simple model and the $1/e$ rule}

Consider now a scenario, where there are $K$ possible protections, where $K$ can be infinite. 
Each protection $j$ is characterized by a cost denoted $x_j>0$ and a
function $s_j(v)$ from $[0,1]$ to $[0,1]$ with the following
interpretation: if the system has a probability of loss $v$ without
the protection $j$, applying the protection $j$ will lower this
probability by a factor of $s_j(v)$ (at a cost $x_j$)

If an agent applies two different protections say $i$ and $j$, then we will assume that
the resulting probability of loss is $s_i(v)s_j(v)$. The rational
behind this assumption is that the protections are independent in a
probabilistic sense. The probability of a successful attack is the
product of the probabilities to elude each of the protections.

For a total budget of $x$, the agent will choose the subset $J\in
[K]=\{1,2,\dots , K\}$ such that $\sum_{j\in J} x_j\leq x$ and which
  minimizes the final probability of loss $\prod_{j\in J}s_j(v)$.
Hence we define the function $p:\mathbb{R}_+\to \mathbb{R}_+$ by,
$p(x,v) =\inf\left\{ \prod_{j\in J}s_j(v)\mbox{ s.t }  \sum_{j\in J} x_j\leq x\right\}$,
so that the optimal security investment problem is still given by (\ref{eq:gl}).
The problem of deriving the function $p(x,v)$ is a standard integer
linear programming problem which can be rewritten as follows $\log p(x,v) =
\inf \left\{\sum_{i\in [K]}e_i\log s_j(v)|\: e_i\in
  \{0,1\}, \sum_{i\in [K]}e_ix_i\leq x \right\}$.

Our aim here is not to address issues dealing with complexity (this problem is known as the knapsack problem) and we
will consider the relaxed problem where $e_i\in [0,1]$. In this case,
the problem is a linear program which is a convex optimization
problem.
The important thing for us is that the function $x\mapsto p(x,v)$ is log-convex
in $x$.
We then have the following generalization of Gordon and Loeb's
Proposition 3:
\begin{theorem}\label{th:1/e}
If the function $x\mapsto p(x,v)$ is non-increasing and log-convex in $x$ then the optimal security investment is bounded by $\ell v/e$.
\end{theorem}
\begin{proof}
We denote $x^*$ the optimal investment and $p^*=p(x^*,v)$, so that
\begin{eqnarray}
\label{eq:alt}\ell p^*+x^*\leq \ell p(x,v)+x.
\end{eqnarray}
We denote $f(x) = \log \ell p(x,v)$.
Firs assume that $x\mapsto p(x,v)$ is continuously differentiable so
that we have
\begin{eqnarray}
\nonumber f(x) &\geq& f(x^*) +f'(x^*)(x-x^*)\\
\label{eq:low}&=& \log \ell p^* -\frac{1}{\ell p^*}(x-x^*),
\end{eqnarray}
where, in the last equality, we used (\ref{eq:1od}). Hence we have,
$f(0) \geq \log\ell p^* +\frac{x^*}{\ell p^*}$,
which can be rewritten as
\begin{eqnarray*}
\ell v \frac{x^*}{\ell p^*} \exp\left(-\frac{x^*}{\ell p^*}\right)\geq x^*.
\end{eqnarray*}
The theorem follows in this case from the observation that $z\exp (-z)\leq e^{-1}$ for $z\geq 0$.

If we do not assume that $x\mapsto p(x,v)$ is continuously differentiable, we will show (\ref{eq:low}) using (\ref{eq:alt}). Namely, suppose there exists $x'\geq 0$ such that
\begin{eqnarray*}
f(x')<\log \ell p^* -\frac{1}{\ell p^*}(x'-x^*).
\end{eqnarray*}
Then by convexity, we have for any $\alpha\in [0,1]$,
\begin{eqnarray*}
f(\alpha x'+(1-\alpha)x^*) &\leq& f(x^*)+\alpha\left( f(x')-f(x^*)\right)\\
&<& \log \ell p^* -\frac{\alpha}{\ell p^*}(x'-x^*).
\end{eqnarray*}
However, by (\ref{eq:alt}), we also have
\begin{eqnarray*}
f(\alpha x'+(1-\alpha)x^*) &\geq & \log\left( \ell p^* -\alpha(x'-x^*)\right)\\
&=& \log\ell p^* -\frac{\alpha}{\ell p^*}(x'-x^*)+O(\alpha^2),
\end{eqnarray*}
and we obtain a contradiction. Hence (\ref{eq:low}) is still true in this case and we can finish the proof as above so that the statement of the theorem holds.
\end{proof}
Theorem \ref{th:1/e} shows that for a broad class of information
security breach probability function, the optimal security investment
is always less than 37\% of the expected loss without protection.
Note that the function $p_{GL}$ introduced above does not satisfy the
conditions of Theorem \ref{th:monot1} but is log-convex so that in
this case, the optimal security investment is always less than 37\% of
the expected loss. Indeed, we saw that for high values of the
vulnerability, the optimal investment is zero. We end this section
with another function $p(x,v)=\frac{v}{(ax+1)^b}$ with
$a,b>0$, which satisfies both
the conditions of Theorems \ref{th:monot1} and \ref{th:1/e}. Hence in
this case, the optimal security investment increases with the
vulnerability but remains below 37\% of the expected loss without protection.

\section{Optimal security investment for an interconnected agent}\label{sec:inter}

We now extend the previous framework in order to model an agent who needs
to decide the amount to spend on security if this agent is part of a
network. In this section, we give results concerning the incentives of an agent in a
network. In the next section, we will consider a security game
associated to this model of agent and determine the equilibrium
outcomes.

\subsection{General model for an interconnected agent}

In order to capture the effect of the network, we will assume that
each agent faces an internal risk and an indirect risk. As explained
in the introduction, the indirect risk takes into account the fact
that a loss can propagate in the network. The estimation of the
internal risk depends only on private information available to the
agent. However in order to decide on the amount to invest in security, the agent needs
also to evaluate the indirect risk. This evaluation depends crucially
on the information on the propagation of the risk in the network available to the
decision-maker. We now describe an abstract and general setting for the
information of the agent.

We assume that the information
concerning the impact of the network on the security of the agent is
captured by a parameter $\gamma$ living in a partially ordered set $\Gamma$
(poset, i.e a set on which there is a binary relation that is
reflexive, antisymmetric and transitive). Indeed this assumption is
not a technical assumption. The interpretation is as follows: $\gamma$
captures the state of the network from the point of view of security
and we need to be able to compare secure states from unsecure ones.

Given $\gamma\in \Gamma$, the agent is able to compute the probability of loss
for any amount $x\in X$ invested in security which is denoted by $p(x,v,\gamma)$.
We still assume that the agent is risk neutral , so that the optimal
security investment is given by:
\begin{eqnarray}
\label{eq:argmin}\varphi(v,\ell,\gamma) =\arg\min\{\ell p(x,v,\gamma)+x:\:x\in X\}.
\end{eqnarray}
 
Note that in our model we consider that
only global statistics about the network are available to all
agents. The state of the network $\gamma$ is public. A 'high' value of
$\gamma$ corresponds to a secure environment, typically with a high
fraction of the population investing in security while a 'low' value
of $\gamma$ corresponds to an unsecure environment with few people
investing in security.
For example, in the epidemic risk model
described below, decision regarding investment are binary and the
public information consists of the parameters of the epidemic risk
model (which are supposed to be fixed) and the fraction $\gamma$ of the population
investing in security. Then for any $\gamma\in [0,1]$, the agent is
able to compute $p(x,v,\gamma)$ as explained below.
Note that in our model, the vulnerability $v$ of an agent is an intrinsic parameter of this agent, in particular it does not depend on the behavior of others or $\gamma$.

\subsection{Epidemic risks model}\label{sec:epidrisk}

In order to gain further insight, we consider in this section the case of economic agents subject to
epidemic risks. This model was introduced in \cite{sig08}. We concentrate here on a simplified version
presented in \cite{lel09}.
In this section, we focus on the dependence of $p(x,v,\gamma)$ in $x$
and $\gamma$. For ease of notation, we remove the explicit dependence
in the vulnerability $v$.

For simplicity, we assume that each agent has a discrete choice regarding self-protection, so that
$X=\{0,1\}$. If she decides to invest in self-protection, we set $x=1$
and say that the agent is in state $S$ as secure, otherwise we set
$x=0$ and say that the agent is in state $N$ as non-secure or negligent.
Note that if the cost of the security product is not one, we can still
use this model by normalizing the loss $\ell$ by the cost of the
security investment.
In order to take her decision, the agent has to evaluate $p(0,\gamma)$ and
$p(1,\gamma)$. To do so, we assume that global statistics on the
network and on the epidemic risks are publicly available and that the
agent uses a simple epidemic model that we now describe.

Agents are represented by vertices of a graph and face two types of
losses: direct and indirect (i.e. due to their neighbors).
We assume that an agent in state $S$ cannot experience a direct
loss and an agent in state $N$ has a probability $p$ of direct loss.
Then any agent experiencing a direct loss 'contaminates' neighbors independently of each
others with probability $q$ if the neighbor is in state $S$ and
$q^+$ if the neighbor is in state $N$, with $q^+\geq q$.
Since only global statistics are available for the graph, we will
consider random families of graphs $G^{(n)}$ with $n$ vertices and
given vertex degree with a typical node having degree distribution
denoted by the random variable $D$ (see \cite{dr07}).
In all cases, we assume that the family of graphs $G^{(n)}$ is
independent of all other processes.
All our results are related to the large population limit ($n$ tends
to infinity). In particular, we are interested in the fraction of the
population in state $S$ (i.e. investing in security) and denoted by 
$\gamma$.

Using this model the agent is able to compute the functions
$p(0,\gamma)$ and $p(1,\gamma)$ thanks to the following result proved
in \cite{sig08} and \cite{netecon08} (using a local mean field):
\begin{proposition}\label{prop:netecon}
Let $\Psi(x)=\mathbb{E}[x^D]$ be the generating function of the degree
distribution of the graph. For any $\gamma\in [0,1]$, there is a unique solution in $[0,1]$ to
the fixed point equation:
$y=1-\gamma \Psi(1-qy)-(1-\gamma)(1-p)\Psi(1-q^+ y)$,
denoted by $y(\gamma)$. Moreover the function $\gamma\mapsto
y(\gamma)$ is non-increasing in $\gamma$. Then we have,
$p(1,\gamma) = 1-\Psi(1-q y(\gamma))$, $p(0,\gamma) = 1-(1-p)\Psi(1-q^+ y(\gamma))$.
\end{proposition}

If we define $h(\gamma) = p(0,\gamma)-p(1,\gamma)$ as the difference of
the two terms given in Proposition \ref{prop:netecon}, we see that the optimal decision is:
\begin{eqnarray}
\label{eq:optic}\ell h(\gamma) > 1 &\Leftrightarrow &\mbox{ agent invests.}
\end{eqnarray}
This equation can be seen as a discrete version of (\ref{eq:1od}). If
the benefit of the protection which is $\ell h(\gamma)$ is more than
its cost (here normalized to one), the agent decides to invest,
otherwise the agents does not invest.
In particular, we observe that the condition for the incentive to invest in
security to increase with the fraction of population investing in
security is given by:
\begin{eqnarray}
\label{eq:h} h(\gamma) = p(0,\gamma)-p(1,\gamma)\mbox{ is an increasing function.}
\end{eqnarray}
We show in the next section that this result extends to a much more
general framework.

Before that, we recall some results of \cite{lel09} describing two
simple cases, one where the condition (\ref{eq:h}) holds and the other
where it does not. The computation presented in this section are done for the standard
Erd\"os-R\'enyi random graphs:
$G^{(n)}=G(n,\lambda/n)$ on $n$ nodes $\{0,1,\dots,n-1\}$, where each
potential edge $(i,j)$, $0\leq i<j\leq n-1$ is present in the graph
with probability $\lambda/n$, independently for all $n(n-1)/2$
edges. Here $\lambda>0$ is a fixed constant independent of $n$ equals
to the (asymptotic as $n\to \infty$) average number of neighbors of an
agent. As explained in the next section, these results extend to a
much more general framework without modifying the qualitative
insights.

We will consider two cases:
\iffalse
\begin{itemize}
\item {\bf Strong protection:} an agent investing in
  protection cannot be harmed at all by the actions or
  inactions of others: $q=0$.% (this is as in \cite{msw06})
\item {\bf Weak protection:} Investing in protection does lower
  the probability of contagion $q$ but it remains positive: $0<q< q^+$.% (as in \cite{sig08})
\end{itemize}
\fi

{\bf Strong protection:} an agent investing in
  protection cannot be harmed at all by the actions or
  inactions of others: $q=0$.
In this case, we have $p(1,\gamma)=0$ so that $h(\gamma) =
p(0,\gamma)$ which is clearly a non-increasing function of $\gamma$ as
depicted on Figure \ref{fig:strong}. 

\begin{figure}[htb]
\begin{center}
\includegraphics[width=4cm]{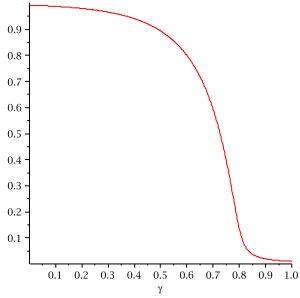}
\hspace{25pt} \caption{Function $h(\gamma)$ for strong
  protection as a function of $\gamma$; $\lambda=10$, $q^+=0.5$, $p=0.01$} \label{fig:strong}
\end{center}\end{figure}
As $\gamma$ the fraction of agents investing in protection
increases, the incentive to invest in protection decreases. In
fact, it is less attractive for an agent to invest in protection,
should others then decide to do so. As more agents invest, the
expected benefit of following suit decreases since there is %a
%reduction in the externalities which translates into 
a lower probability of loss, the network becoming more secure.

{\bf Weak protection:} investing in protection does lower
  the probability of contagion $q$ but it remains positive: $0<q< q^+$.
In this case, the map $\gamma \mapsto h(\gamma)$ can be non-decreasing for
small value of $\gamma$ and decreasing for values of $\gamma$ close to
one (see Figure \ref{fig:weak}).
\begin{figure}[htb]
\begin{center}
\includegraphics[width=4cm]{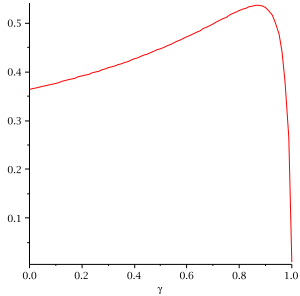}
\hspace{25pt} \caption{Function $h(\gamma)$ for weak
  protection as a function of $\gamma$; $\lambda=10$, $q^+=0.5$, $p^+=0.01$ and $q=0.1$}\label{fig:weak}
\end{center}\end{figure}
For small values of $\gamma$, the incentive for an agent to invest in
security actually increases with the proportion of agents investing in
security (recall Condition (\ref{eq:h})). We will see in the next
section, that this alignment of incentives is responsible for a
coordination problem when agents are strategic.
%In other words, an agent benefits from the investment in security of
%the other agents and for small values of $\gamma$ the agent prefers
%to 'free-ride' the investment of the other agents.

\subsection{Sufficient conditions for monotone investment in a
  network}

We now show how the condition (\ref{eq:h}) extends to a general
framework. This extension is given by the following result:
\begin{theorem}\label{th:gen}
If the function $p(x,v,\gamma)- p(x,v',\gamma')$ is strictly
increasing in $x\in X$ for any $(v',\gamma')>(v,\gamma)$ and the
function $p(x,v,\gamma)$ is non-increasing in $x$,
then $\varphi(v,\ell,\gamma)$ defined in (\ref{eq:argmin}) is
non-decreasing.
\end{theorem}
\begin{proof}
As noticed in Remark \ref{rem:sub}, we need to prove that our
condition ensures that $\ell p(x,v,\gamma)-\ell' p(x,v',\gamma')$ is strictly
increasing in $x\in X$ for any $(v',\ell',\gamma')>(v,\ell,\gamma)$.
If $\ell=\ell'$, this follows from the condition of the theorem. We
now deal with the case $\ell'>\ell$. Let $x'>x$, then by the condition
of the theorem, we have
\begin{eqnarray*}
\ell p(x',v,\gamma)-\ell p(x',v',\gamma') > \ell p(x,v,\gamma)-\ell p(x,v',\gamma'),
\end{eqnarray*}
but since $\ell'>\ell$ and $p(x,v',\gamma')-p(x',v',\gamma')\geq 0$ for $x'>x$, we also have
\begin{eqnarray*}
\ell p(x',v',\gamma')-\ell'p(x',v,',\gamma') > \ell p(x,v',\gamma')-\ell' p(x,v',\gamma').
\end{eqnarray*}
Summing these inequalities gives exactly the desired result.
\end{proof}
\begin{remark}
%\begin{itemize}
%\item 
%Theorem \ref{th:monot1} can be obtained from Theorem \ref{th:gen} by fixing the parameter $\gamma$.
%\item 
Clearly, the condition of Theorem \ref{th:gen} translates in the
setting described in Section \ref{sec:epidrisk} to
$p(0,\gamma)-p(0,\gamma')< p(1,\gamma)-p(1,\gamma'),\mbox{ for any }\gamma'>\gamma$,
which corresponds exactly to (\ref{eq:h}).
%\end{itemize}
\end{remark}
In the particular case where $\Gamma$ is a subset of $\mathbb{R}$, and
under some smoothness conditions, we obtain:

\begin{proposition}\label{prop:3}
If the function $p(x,v,\gamma)$ is twice continuously differentiable
on $X\times [0,1]\times \Gamma$, then sufficient conditions for
$\varphi(v,\ell,\gamma)$ to be non-decreasing are:
$\frac{\partial p}{\partial x}(x,v,\gamma) \leq 0$, $\frac{\partial^2
  p}{\partial x\partial v}(x,v,\gamma)\leq 0$, $\frac{\partial^2
  p}{\partial x\partial \gamma} (x,v,\gamma)\leq 0$.
\end{proposition}

As we will see in the next section satisfying the conditions of Theorem \ref{th:gen} (or
Proposition \ref{prop:3}) ensures that the incentives in the
population are aligned but this might lead to a coordination problem.

\section{Equilibrium analysis of the security game }\label{sec:equ}

We now present our results in a game-theoretic framework where each
agent is strategic.
We assume that the effect of the action of any single agent is infinitesimal but each agent anticipates the effect of the actions of all other agents on the security level of the network.

\subsection{Information structure and fulfilled expectations equilibrium}

In most of the literature on security games, it is assumed that the
player has complete information.
In particular, each player knows the probability of propagation of the
attack or failure from each other player in the network and also the cost
functions of other players. In this case, the agent is able to
compute the Nash equilibria of the games (if no constraint is made on
the computing power of the agent) and decides on her level of
investment accordingly. In particular, the agent is able to solve
(\ref{eq:argmin}) for all possible values of $\gamma$ which capture the
decision of all other agents. Note that even if only binary decisions are made
by agents the size of the set $\Gamma$ grows exponentially with the
number of players in the network. Moreover in a large network, the
complete information assumption seems quite artificial, especially for
security games where complete information would then implies that the
agents disclose information on their security strategy to the public
and hence to the potential attacker!

Here we relax the assumption of complete information. As in previous
section, we assume that each agent is able to compute the function
$p(x,v,\gamma)$ based on public information and on the epidemic risk
model. The values of the possible loss $\ell$ and the vulnerability $v$ are private information of the
agent and vary among the population.
In order to define properly the equilibrium of the game, we
assume that all players are strategic and are able to do this
computation. Hence if a player expect that a fraction $\gamma^e$ of
the population invests in security, she can decide for her own
investment.
We assume that at equilibrium expectations are fulfilled so that at
equilibrium the actual value of $\gamma$ coincides with $\gamma^e$.
This concept of fulfilled expectations equilibrium to model network externalities is standard in economics (see Section 3.6.2 in \cite{fk07}).
%was introduced by Katz and Shapiro in \cite{ks85}.

We now describe it in more details. For simplicity of the
presentation, we do not consider the dependence in the vulnerability
$v$ since in the security game, we focus on the monotonicity in $\gamma$ which will turn out to be crucial.
We also consider that the choice regarding investment is binary,
i.e. $X=\{0,1\}$.

We consider a heterogeneous population, where agents differ in loss
sizes only. This loss size $\ell$ is called the type of the agent.
We assume that agents
expect a fraction $\gamma^e$ of agents in state $S$, i.e. to make
their choice, they assume that the fraction of agents investing in
security is $\gamma^e$.
We now define a network externalities function that captures the
influence of the expected fraction of agents in state $S$ on the
willingness to pay for security.
Let the network externalities function be $h(\gamma^e)$. More
precisely, for an agent of type $\ell$, the willingness to pay for
protection in a network with a fraction $\gamma^e$ of the agents
in state $S$ is given by $\ell h(\gamma^e)$ so that if
\begin{eqnarray}
\label{eq:hy}\ell h(\gamma^e)\geq c, \mbox{ (where $c$ is the cost of
  the security option)}
\end{eqnarray}
the agent will invest and otherwise not. Hence
(\ref{eq:hy}) is in accordance with (\ref{eq:optic}) (where the cost
was normalized to one). Note that here, we
do not make any a priori assumption on the network externalities function $h$
which can be general and fit to various models.

Indeed, our model corresponds
exactly to the multiplicative formulation of Economides and Himmelberg
\cite{ecohim95} which allows different types of agents to receive
differing values of network externalities from the same network.
As explained above, agents with low $\ell$ have little or no use for the
protection whereas agents with high $\ell$ value highly
security. This is taken into account in our model since for a fixed
expected fraction of agents in state $S$, agents with high $\ell$ have
a higher willingness to pay for self-protection than agents with low
$\ell$.

Let the cumulative distribution function of types be $F(\ell)$, i.e the
fraction of the population having type lower than $\ell$ is given by
$F(\ell)\leq 1$.
We make the following hypothesis:
\begin{hypo}\label{hypoF}
$F(\ell)$ is continuous with positive density everywhere
on its support which is normalized to be $[0,1]$. 
\end{hypo}
Note in particular that $F$ is strictly increasing and it
follows that the inverse $F^{-1}(\gamma)$ is well-defined for
$\gamma\in[0,1]$.

Given expectation $\gamma^e$ and cost for protection $c$, all agents with type $\ell$
such that $\ell h(\gamma^e)> c$ will invest in protection. Hence
the actual fraction of agents investing in protection is given by
%\begin{eqnarray}
$\gamma = 1-F\left(\min\left(\frac{c}{h(\gamma^e)},1\right)\right)$.
%\end{eqnarray}
Hence following
\cite{ecohim95}, we can invert this equation and we define the willingness to pay for the
last agent in a network of size $\gamma$ with expectation
$\gamma^e$ as
\begin{eqnarray}
\label{eq:wtp}w(\gamma,\gamma^e) = h(\gamma^e)F^{-1}(1-\gamma).
\end{eqnarray}
Seen as a function of its first argument, this is just an inverse
demand function: it maps the quantity of protection demanded to
the market price.
Because of externalities, expectations affect the
willingness to pay:
\begin{eqnarray}
\label{eq:d2wtp}\frac{\partial w}{\partial \gamma^e}(\gamma,\gamma^e) =
h'(\gamma^e)F^{-1}(1-\gamma).
\end{eqnarray}

For goods that do not exhibit network externalities, demand slopes
downward: as price increases, less of the good is demanded. This
fundamental relationship may fail in goods with network externalities.
If $h'(.)>0$, then the
willingness to pay for the last unit may increase as the number expected
to be sold increases as can be seen from (\ref{eq:d2wtp}):
$\frac{\partial w}{\partial \gamma^e}(\gamma,\gamma^e)>0$.
For example in \cite{ecohim95} studying the FAX market, as more and more
agents buy a FAX, the utility of the FAX increases since more and more
agents can be reached by this communication mean.
%In this case, if the expected fraction $\gamma^e$ rises with the actual $\gamma$, then we see that the willingness to pay given by (\ref{eq:wtp}) may increase with the fraction of population investing.
For a fixed cost $c$, in equilibrium, the expected
fraction $\gamma^e$ and the actual one $\gamma$ must satisfy
\begin{eqnarray}
\label{eq:eq}c =w(\gamma,\gamma^e) = h(\gamma^e)F^{-1}(1-\gamma).
\end{eqnarray}
If we assume moreover that in equilibrium, expectations are fulfilled, then the possible equilibria are given by the fixed point equation:
\begin{eqnarray}
c = w(\gamma,\gamma) = h(\gamma)F^{-1}(1-\gamma)=:w(\gamma).
\end{eqnarray}
We see that if $h'(.)>0$, the concept of fulfilled expectations
equilibrium captures the possible increase in the willingness to pay
as the number expected to be sold increases. This would corresponds to
the case where we have $w'(\gamma)>0$ for some values of $\gamma$.
In such cases, a critical mass phenomenon (as in the FAX market \cite{ecohim95}) can occurs
: there is a problem of coordination. 
We explain this phenomenon more formally in the next section and then show how our results differ from \cite{ecohim95}.
We end this section with the following important remark:
\begin{remark}\label{rem:unif}
The case of an homogeneous population in which all agents have the
same type, i.e the same loss size $\ell$ corresponds to the function
$F^{-1}$ being constant equal to $\ell$. In this case, the willingness
to pay is simply $w(\gamma)=h(\gamma)\ell$. In particular, the
epidemic risk model presented above can be used to model the network
externalities by the function $h(\gamma)$ computed in Section
\ref{sec:inter}. In this case, Condition (\ref{eq:h}) still gives a
condition for incentives to be align. As we will see next, this
condition might lead to  critical mass: if incentives are aligned, there is a
coordination problem!
\end{remark}

\subsection{Critical mass: coordination problem}

To determine the possible equilibria, we analize the shape of the
fulfilled expectations demand $w(\gamma)$. First we have $w(0)\geq 0$
which is equal to the value of the self-protection assuming there are
no network externalities. We also have $w(1)=0$ since by Hypothesis
\ref{hypoF}, we have $F^{-1}(0)=0$. In words, this means that there
are agents with very low $\ell$ who have little or no interest in
self-protection. Then in order to secure completely the network, we
have to convince even agents of very low willingness to pay.

The slope of the fulfilled expectations demand is
\begin{eqnarray}
\label{eq:p'}w'(\gamma) =
-\frac{h(\gamma)}{F'(F^{-1}(1-\gamma))}+h'(\gamma)F^{-1}(1-\gamma).
\end{eqnarray}

The first term measures the slope of the inverse demand without taking
into account the effect of the expectations. The second term
corresponds to the effect of an increase in the expected fraction of
agents in state $S$. If $h'(.)>0$ as in \cite{ecohim95}, it
corresponds to the increase in the willingness to pay of the last
agent investing in self-protection created by his own action in
joining the group of agents in state $S$.
Note that in any case, if the fraction of agents in state $S$ gets
very large, i.e. $\gamma\to 1$, the last agent investing in
self-protection has very low willingness to pay for it. Hence for
$\gamma$ close to one, the effect of marginal expectations on the
marginal agent investing in $S$ is negligible. Formally this is
observed by $\lim_{\gamma\to 1}h'(\gamma)F^{-1}(1-\gamma)=0$.
It follows that 
\begin{eqnarray}
\label{eq:p'1}\lim_{\gamma\to 1}w'(\gamma) = \lim_{\gamma\to
  1}-\frac{h(\gamma)}{F'(F^{-1}(1-\gamma))}=-\frac{h(1)}{F'(0)}<0.
\end{eqnarray}
Note that we allow $F'(0)=0$ in which case, Equation (\ref{eq:p'1})
should be interpreted as $\lim_{\gamma\to 1}w'(\gamma) = -\infty$.
The sign of $\lim_{\gamma\to 0}w'(\gamma)$ depends on the parameters
of the model and we will see that it is of crucial importance. 
We make the following hypothesis
\begin{hypo}\label{hypow}
The function $w(\gamma)$ is single-peaked.
\end{hypo}
Note that in the case of an homogeneous population, $w(\gamma)=h(\gamma)\ell$, where $h(\gamma)$ was computed in Section \ref{sec:inter} for the epidemic risk model and is single-peaked. 

We are now ready to state the main result of this section:
\begin{theorem}\label{th:comp}
Under Hypothesis \ref{hypoF} and \ref{hypow}, a network has positive critical mass if $\lim_{\gamma\to
  0}h'(\gamma)>0$ and either
\begin{itemize}
\item[(i)] $w(0)=0$, i.e. if all agents are in state $N$ then no agent
  is willing to invest in self-protection;
\item[(ii)] $\lim_{\gamma\to 0}h'(\gamma)$ is sufficiently large,
  i.e. there are large private benefits to join the group of agents in
  state $S$ when the size of this group is small;
\item[(iii)] $\lim_{\gamma\to 1}F'(\gamma)$ is sufficiently large,
  i.e. there is a significant density of agents who are ready to
  invest in self-protection even if the number of agents already in
  state $S$ is small.
\end{itemize}
\end{theorem}
\begin{remark}
Note that if $h'(\gamma)>0$ for small values of $\gamma$, then
incentives are aligned by results of previous Section but this might
lead to a coordination problem. Indeed as shown by previous theorem,
this is a necessary condition for a network to exhibit positive
critical mass.
In the case of a homogeneous population (see Remark \ref{rem:unif}), the function $w(\gamma)$ is proportional to the function $h(\gamma)$ computed in Section \ref{sec:inter} for the epidemic risk model. In particular, in the case of weak protection, there is positive critical mass as shown by Figure \ref{fig:weak}.
\end{remark}
\begin{proof}
Since we proved that $\gamma \mapsto w(\gamma)$ is decreasing for
$\gamma$ close to one, there are only two possibilities: either is is
increasing for small values of $\gamma$ or it is decreasing for all
$\gamma$.
As explained in Lemma 1 of \cite{ecohim95}, the network has a positive critical
mass if and only if $\gamma\mapsto w(\gamma)$ is increasing for small
values of $\gamma$.

\begin{figure}[htb]
\begin{center}
\psfrag{pi}{$c$}
\psfrag{pi0}{$c_0$}
\psfrag{gam}{$\gamma$}
\psfrag{gami}{$\gamma_i$}
\psfrag{gams}{$\gamma_s$}
\psfrag{gam0}{$\gamma_0$}
\includegraphics[angle=0,width=4cm]{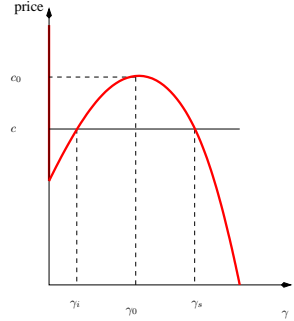}
\hspace{25pt} \caption{Willingness to pay curve (or
  demand curve) $w(\gamma)$} \label{fig:dem}
\end{center}\end{figure}
This is illustrated thanks to Figure \ref{fig:dem} (which should be
compared to Figure \ref{fig:weak}).
Recall that in equilibrium, we have
$w(\gamma^*)=h(\gamma^*)F^{-1}(1-\gamma^*)=c$.
If we imagine a
constant cost $c$ decreasing parametrically, the network will start at
a positive and significant size $\gamma^0$ corresponding to a cost
$c^0$. For each smaller cost $c^1<c<c^0$, there are three values of $\gamma^*$
consistent with $c$: $\gamma^*=0$; an unstable value of $\gamma^*$ at
the first intersection of the horizontal through $c$ with
$w(\gamma)$; and the Pareto optimal stable value of $\gamma^*$ at the
largest intersection of the horizontal with $w(\gamma)$.

As explained above, a network exhibits a positive critical mass if and
only if $\lim_{\gamma\to 0}w'(\gamma)>0$. Now by (\ref{eq:p'}), we
have
$\lim_{\gamma\to 0}w'(\gamma) = \lim_{\gamma\to 0}h'(\gamma)
-\frac{h(0)}{\lim_{\gamma\to 1}F'(\gamma)}$,
note that $h(0)=w(0)$ and the theorem follows easily.
\end{proof}

We finish this section by explaining the main difference
between our model and models with standard positive
externalities. 
Informally, in the model of \cite{ecohim95} for the FAX market, when a
new agent buys the good
(a FAX machine), he has a personal benefit and he also increases the
value of the network of FAX machines. This is a positive externality
which are felt only by the adopters of the good. 
Indeed, when this agent buys the good, this is a negative externality
on the agents who did not buy the good (see \cite{ks85}, Example A9 in \cite{segal}).
In our case, when an agent
chooses to invest in security, the externalities are always positive
and we have to distinguish between two positive externalities: one is
felt by the agents in state $S$ and the other is felt by the agent in state $N$. 
Indeed as $\gamma$ increases, both populations experience a decrease of their probability of loss but the value of this decrease is not the same in both populations.
We call the 'public externalities' the decrease felt by agents in state $N$ and it is given by $g(\gamma)=p(0,0)-p(0,\gamma)\geq 0$. We call the 'private externalities' the decrease felt only by agents in state $S$ and it is given by
$g(\gamma)+h(\gamma)=p(0,0)-p(1,\gamma)\geq g(\gamma)$.
%, where $h(\gamma)\geq 0$.

First note that the notations are consistent. In particular, Equation
(\ref{eq:hy}) still gives the willingness to pay for self-protection
in a network with a fraction $\gamma^e$ of the agents in state $S$.
We are still dealing with positive externalities, however this does
not imply that $h'(.)>0$ (as it is the case in \cite{ecohim95}). Instead, positive externalities (i.e. the fact that both the public externalities $g(\gamma)$ and the private externalities $g(\gamma)+h(\gamma)$ are increasing in $\gamma$) only ensures that:
\begin{eqnarray}
\label{eq:mon}g'(.)\geq 0 \mbox{ and, } g'(.)+h'(.)\geq 0.
\end{eqnarray}
Assumption (\ref{eq:mon}) ensures the sensible fact that the more
agents invest in self-protection, the more secure the network becomes
(this is the total effect).
If in addition, $h'(.)\geq 0$, then adoption of security increases
others' incentive to invest (this is the marginal effect) and there
might be a critical mass effect. Recent works on the marginal effect
include Segal's increasing externalities \cite{segal} or
Topkis'supermodularity \cite{topbook}.
On the contrary when $h'(.)<0$, there is no
coordination problem (no critical mass). However, we show in the next
section that even in this case, the equilibrium is not socially
efficient. The intuition for this fact is that incentives are not
anymore aligned and since agent benefits from the
investment in security of the other agents, they prefers to
'free-ride' the investment of the other agents.
%If no agent did invest in self-protection, then there is no public externalities hence we have: $ g(0)=0$.

\subsection{Welfare Maximization}

A planner who maximizes social welfare can fully internalize the
network externalities and this is the situation we now consider.
We will show that there is always efficiency loss in our model with
exogenous price. In other words, the price of anarchy is always
greater than one. 
\begin{theorem}
Under Hypothesis \ref{hypoF} and \ref{hypow}, a social planner will choose a larger fraction $\gamma$ of agents
investing in self-protection than the market equilibrium for any fixed cost $c$.
\end{theorem}
We refer to \cite{sig08} for an estimate of this price of
anarchy for the epidemic risks model presented in previous section and to \cite{netecon08} for an extension to graphs with power-law degrees distribution.
\begin{proof}
The social welfare function is:
\begin{eqnarray*}
W(\gamma) &=& g(\gamma)\int_\gamma^1F^{-1}(1-u)du\\
%\left( g(\gamma)+h(\gamma)\right)\int_0^\gamma F^{-1}(1-u)du-c\gamma,
&+& \left( g(\gamma)+h(\gamma)\right)\int_0^\gamma F^{-1}(1-u)du-c\gamma,
\end{eqnarray*}
where $g(\gamma)\int_\gamma^1F^{-1}(1-u)du$ is the gross benefit for
the fraction of agents in state $N$ and $\left(
  g(\gamma)+h(\gamma)\right)\int_0^\gamma F^{-1}(1-u)du$ for the
fraction of agents in state $S$ and $c\gamma$ are the costs.
We denote by $B(\gamma)$ the gross benefit for the
whole population so that $W(\gamma)=B(\gamma)-c\gamma$, then we
have:
\begin{eqnarray*}
B'(\gamma) &=& h(\gamma)F^{-1}(1-\gamma)\\%+\left( h'(\gamma)+g'(\gamma)\right)\int_0^\gamma F^{-1}(1-u)du \\
%&&+\:g'(\gamma)\int_\gamma^1F^{-1}(1-u)du.%\\
&=& +\left( h'(\gamma)+g'(\gamma)\right)\int_0^\gamma F^{-1}(1-u)du\\ &+&g'(\gamma)\int_\gamma^1F^{-1}(1-u)du.
\end{eqnarray*}
Recall that by (\ref{eq:eq}), the equilibria of the game (without the
social planner) are the values $\gamma$ such that $w(\gamma)=h(\gamma)F^{-1}(1-\gamma)=c$. In
particular for such a value of $\gamma$, since we assume positive externalities (\ref{eq:mon}), we have
that $B'(\gamma)\geq w(\gamma)=c$, hence $W'(\gamma)\geq 0$ and the theorem follows.
\end{proof}

\section{Conclusion}

In this paper, we study under which conditions agents in a large network invest in self-protection. We started our analysis with finding conditions when the amount of investment increases for a single agent as the vulnerability and loss increase. We also showed that risk-neutral agent do not invest more than 37\% of the expected loss under log-convex security breach probability functions.
We then extended our analysis to the case of interconnected agents of a large network using a simple epidemic risk models. We derived a sufficient condition on the security breach probability functions taking into consideration the global knowledge on the security of the entire network for guaranteeing increasing investment with increasing vulnerability.
It would be interesting to use other epidemics models as in \cite{leldiff} to see the impact on the results of this section.

Finally, we study a security game where agents anticipate the effect
of their actions on the security level of the network. We showed that
in all cases, the fulfilled equilibrium is not socially efficient. We
explained it by the separation of the network externalities in two
components: one public (felt by agents not investing) and the other
private (felt only by agents investing in self-protection). 
We also showed that alignment of incentives typically leads to a coordination problem.
%In the case of an inhomogeneous population, we showed that the condition derived (in Theorem \ref{th:gen}) to ensure the monotinicity of investment of an interconnected agent with respect to the global security level of the netwok, also ensures that there is no coordination problem.

In view of our results, it would be interesting to derive sufficient
conditions for non-alignment of the incentives as these conditions
would ensure that there is no coordination problem.
%Note that it might be that this sufficient condition is a too strong requirement in order to avoid coordination problem. 
Exploring this issue is an interesting open problem. Another interesting direction of research concerns the information structure of such games. For example, in the case presented here of epidemic risk model, what is the impact of an error in the estimation of the contagion probability which could be for example over evaluated by the firm selling the security solution?
Also, in our work, the attacker is not a strategic player: attacks are made at random with probability of success depending of the security level of the agent targeted. However if the attacker can observe the security policies taken by the defenders, it can exploit this information \cite{worm}. An interesting extension would be to incorporate in our model such a strategic attacker as in \cite{goyal}. Another extension could also consider the supply side, i.e. the firms distributing the security solution in the population. Very basic cases have been studied \cite{lelsig09,BCGS} but again with a non strategic attacker.


\section*{Acknowledgements}
The author acknowledges the support of the French Agence Nationale de la Recherche (ANR) under reference ANR-11-JS02-005-01 (GAP project).



%\bibliographystyle{IEEEtran}
%\bibliography{eco}

\end{document}